\documentclass[a4paper]{article}
\pdfoutput=1 

\usepackage[l2tabu, orthodox]{nag}
\usepackage[utf8]{inputenc}
\usepackage{csquotes}
\usepackage[english]{babel}
\usepackage{hyperref}
\usepackage{amsmath}
\usepackage{amssymb}
\usepackage{amsthm}
\usepackage{bm}
\usepackage{authblk}

\usepackage[isbn=true, url=true, maxnames=20, style=numeric,]{biblatex}
\newtheorem{theorem}{Theorem}[section]

\newtheorem{corollary}[]{Corollary}[section]

\DeclareMathOperator*{\rank}{rank}

\DeclareMathOperator*{\im}{im}
\DeclareMathOperator*{\coker}{coker}

\newcommand\restr[2]{{
  \left.\kern-\nulldelimiterspace 
  #1 
  \vphantom{\big|} 
  \right|_{#2} 
  }}

\usepackage{mathtools}
\mathtoolsset{
showonlyrefs=true 
}

\makeatletter
\newcommand\Autoref[1]{\@first@ref#1,@}
\def\@throw@dot#1.#2@{#1}
\def\@set@refname#1{
    \edef\@tmp{\getrefbykeydefault{#1}{anchor}{}}%
    \def\@refname{\@nameuse{\expandafter\@throw@dot\@tmp.@autorefname}s}%
}
\def\@first@ref#1,#2{%
  \ifx#2@\autoref{#1}\let\@nextref\@gobble
  \else%
    \@set@refname{#1}
    \@refname~\ref{#1}
    \let\@nextref\@next@ref
  \fi%
  \@nextref#2%
}
\def\@next@ref#1,#2{%
   \ifx#2@ and~\ref{#1}\let\@nextref\@gobble
   \else, \ref{#1}
   \fi%
   \@nextref#2%
}
\makeatother

\title{SMSSVD -- SubMatrix Selection Singular Value Decomposition}
\author[1,2]{Rasmus Henningsson}
\author[1,2,3,4]{Magnus Fontes\thanks{fontes@maths.lth.se; Corresponding author}}
\affil[1]{{\small The Centre for Mathematical Sciences, Lund University, Sweden}}
\affil[2]{{\small The International Group for Data Analysis, Institut Pasteur, Paris, France}}
\affil[3]{{\small The Center for Genomic Medicine, Rigshospitalet, Copenhagen, Denmark}}
\affil[4]{{\small Persimune, The Centre of Excellence for Personalized Medicine, Copenhagen, Denmark}}

\begin{document}

\maketitle

\begin{abstract}
High throughput biomedical measurements normally capture multiple overlaid biologically relevant signals and often also signals representing different types of technical artefacts like e.g. batch effects.
Signal identification and decomposition are accordingly main objectives in statistical biomedical modeling and data analysis.
Existing methods, aimed at signal reconstruction and deconvolution, in general, are either supervised, contain parameters that need to be estimated or present other types of ad hoc features.
We here introduce SubMatrix Selection SingularValue Decomposition (SMSSVD), a parameter-free unsupervised signal decomposition and dimension reduction method, designed to reduce noise, adaptively for each low-rank-signal in a given data matrix, and represent the signals in the data in a way that enable unbiased exploratory analysis and reconstruction of multiple overlaid signals, including identifying groups of variables that drive different signals.

The Submatrix Selection Singular Value Decomposition (SMSSVD) method produces a denoised signal decomposition from a given data matrix.
The SMSSVD method guarantees orthogonality between signal components in a straightforward manner and it is designed to make automation possible.
We illustrate SMSSVD by applying it to several real and synthetic datasets and compare its performance to golden standard methods like PCA (Principal Component Analysis) and SPC (Sparse Principal Components, using Lasso constraints).
The SMSSVD is computationally efficient and despite being a parameter-free method, in general, outperforms existing statistical learning methods. 

A Julia implementation of SMSSVD is openly available on GitHub (\href{https://github.com/rasmushenningsson/SMSSVD.jl}{https://github.com/rasmushenningsson/SMSSVD.jl}).
\end{abstract}

\section{Introduction}

High throughput biomedical measurements, by design, normally capture multiple overlaid biologically relevant signals, but often also signals representing different types of biological and technical artefacts like e.g. batch effects.
There exist different methods aimed at signal reconstruction and deconvolution of the resulting high dimensional and complex datasets, but these methods almost always contain parameters that need to be estimated or present other types of ad hoc features.
Developed specifically for Omics data and more particularly gene expression data such methods include the gene shaving method \cite{hastie2000gene}, tree harvesting \cite{hastie2001supervised}, supervised principal components \cite{bair2004semi} and amplified marginal eigenvector regression \cite{ding2017predicting}.
They employ widely different strategies do deal with the ubiquitous $P\gg N$ (many more variables than samples) problem in omics data.
Gene Shaving uses the first principal component to iteratively guide variable selection towards progressively smaller nested subsets of correlated genes with large variances.
An optimal subset size is then chosen using the `gap statistic', a measure of how much better the subset is than what is expected by random chance.
To find additional subsets (signals), each gene is first projected onto the orthogonal complement of the average gene in the current subset, and the whole process is repeated.

We here introduce SubMatrix Selection Singular Value Decomposition (SMS\-SVD), a parameter-free unsupervised dimension reduction technique primarily designed to reduce noise, adaptively for each low-rank-signal in a data matrix, and represent the data in a way that enable unbiased exploratory analysis and reconstruction of the multiple overlaid signals, including finding the variables that drive the different signals.

Our first observation for the theoretical foundation of SMSSVD is that the SVD of a linear map restricted to a hyperplane (linear subspace) share many properties with the SVD of the corresponding unrestricted linear map.
Using this we show that, by iteratively choosing orthogonal hyperplanes based on criteria for optimal variable selection and concatenating the decompositions, we can construct a denoised decomposition of the data matrix.
The SMSSVD method guarantees orthogonality between components in a straightforward manner and coincide with the SVD if no variable selection is applied.
We illustrate the SMSSVD by applying it to several real and synthetic datasets and compare its performance to golden standard methods for unsupervised exploratory analysis: Classical PCA (Principal Component Analysis) \cite{hotelling1933analysis} and the lasso or elastic net based methods like SPC (Sparse Principal Components) \cite{witten2009penalized}.
Just like PCA and SPC, SMSSVD is intended for use in wide range of situtations, and no assumptions specific to gene expression analysis are made in the derivation of the method.
The SMSSVD is computationally efficient and despite being a parameter-free method, in general, it outperforms or equals the performance of the golden standard methods. A Julia implementation of SMSSVD is openly available on GitHub.

\section{Methods}

\begin{theorem}
\label{thmSubspaceSVD}
Let $\restr{X}{\Pi}: \Pi\to X(\Pi)$ be the restriction of a linear map $X:\mathbb{R}^N \to \mathbb{R}^P$ to a $d$-dimensional subspace $\Pi\subset\mathbb{R}^N$ such that $\Pi\perp\ker X$.
Furthermore, let $U\Sigma V^T = \sum_{i=1}^d\sigma_iU_{\cdot i}V_{\cdot i}^T$ be the singular value decomposition of $\restr{X}{\Pi}$.
Then
\begin{enumerate}
	\item $V_{\cdot i} \perp \ker X, \;\forall i$.
	\item $U_{\cdot i} \perp \coker X, \;\forall i$.
	\item $XV = U\Sigma$.
	\item $U^TX = \Sigma V^T + U^TX(I-VV^T)$.
	\item $(I-UU^T)X(I-VV^T) = (I-UU^T)X$.
	\item $\rank X = d + \rank (I-UU^T)X$.
	\item $\rank\left(X\right) = d + \rank\left((I-UU^T)X\right)$.
\end{enumerate}
\end{theorem}
Remark. {\textit In the statement of the theorem and in the proof below, we consider all vectors to belong to the full-dimensional spaces.
In particular, we extend all vectors in subspaces of the full spaces with zero in the orthogonal complements.}
\begin{proof}
\textit{1.} The columns of $V$ are an orthonormal basis of $\Pi$ and thus orthogonal to $\ker X$.
\textit{2.} The columns of $U$ are an orthonormal basis of $X(\Pi)$ and $X(\Pi) \perp \coker X$.
\textit{3.} $XV=\restr{X}{\Pi}V=U\Sigma V^TV=U\Sigma$.
\textit{4.} Using \textit{3} we get
\begin{align*}
U^TX &= U^TXVV^T + U^TX(I-VV^T)\\
     &= \Sigma V^T + U^TX(I-VV^T).
\end{align*}
\textit{5.} The statement follows from $(I-UU^T)XV = (I-UU^T)U\Sigma = \bm{0}$, where we have used that $U^TU=I$.
\textit{6.}
Let $Y\coloneqq X(\Pi)$ and $Z\coloneqq \im X / X(\Pi)$ be the parts of the decomposition $\im X = Y \oplus Z$, which is possible since $Y\subset \im X$.
The linear map $(I-UU^T)$ is orthogonal projection onto $X(\Pi)^\perp$ and thus maps $Y \to 0$ and $Z \to Z$.
Since $\rank A = \dim\left(\im A\right)$, it follows immediately that $\rank (I-UU^T)X = \dim Z$ and that $\rank X = \dim Y + \dim Z = d + \dim Z$.
\end{proof}
Note that $V^TV$ is the orthogonal projection on $\Pi$ and $U^TU$ is the orthogonal projection on $X(\Pi)$.
If $\Pi$ is spanned by the right singular vectors corresponding to the $d$ largest singular values of $X$, then $U\Sigma V^T$ is the truncated SVD which by the \textit{Eckhart-Young Theorem} is the closest rank $d$ matrix to $X$ in Frobenius and Spectral norms.
Furthermore, if $\Pi=(\ker X)^\perp$, then $d = \rank X$ and $U\Sigma V^T$ is the SVD of $X$ (without expanding $U$ and $V$ to orthonormal matrices).
Also note that for these two cases, property \textit{4} takes a simpler form, $U^TX=\Sigma V^T$ (symmetric to property \textit{3}), but the residual $U^TX(I-VV^T)$ is nonzero in general.

\autoref{thmSubspaceSVD} concerns the relationship between $X$ and $U\Sigma V^T$ and shows that many important properties that hold for the (truncated) SVD are retained regardless how the subspace $\Pi$ is chosen.
The results from \autoref{thmSubspaceSVD} are put into practice in this iterative algorithm.
Let $X_1 \coloneqq X$ and repeat the following steps for $k=1,2,\hdots$
\begin{enumerate}
	\item Choose $\Pi_k$.
	\item Compute $U_k\Sigma_kV_k^T$ from $\restr{X_k}{\Pi_k}$.
	\item Let $X_{k+1} \coloneqq (I-U_kU_k^T)X_k$.
\end{enumerate}
The iterations can continue as long as $X_k$ is nonzero or until some other stopping criteria is met.
Finally, the results are concatenated:
\begin{align}
U\Sigma V^T &\coloneqq
\begin{pmatrix}U_1 & U_2 & \hdots & U_n\end{pmatrix}
\begin{pmatrix} \Sigma_1 &          &         &          \\
                         & \Sigma_2 &         &          \\
                         &          & \ddots  &          \\
                         &          &         & \Sigma_n \end{pmatrix}
\begin{pmatrix}V_1^T \\ V_2^T \\ \vdots \\ V_n^T\end{pmatrix}\\
&= \sum_{k=1}^n U_k\Sigma_kV_k^T.
\end{align}
Orthogonality between columns within each $U_k$ and $V_k$ respectively follow immediately from the definition.
Step 3 above, together with \autoref{thmSubspaceSVD}, property \textit{2}, guarantees orthogonality between the columns of different $U_k$'s, since the columns of $U_k$ are in $\coker X_l$, for all $l>k$.
Similarly, properties \textit{5} and \textit{1} of \autoref{thmSubspaceSVD} imply orthogonality between the columns of different $V_k$'s.
That is, $U^TU=V^TV=I$.
The diagonal entries of each $\Sigma_k$ are decreasing, but the algorithm above does not ensure any structure between the blocks.
In practice however, with each $\Pi_k$ chosen to capture a strong signal in $X_k$, we can expect the SMS singular values to be decreasing, or at least close to decreasing.

The rank decreases by $d_k$ in each iteration, that is $\rank X_k = d_k + \rank X_{k+1}$, which follows from property \textit{6} in \autoref{thmSubspaceSVD}.
This implies that $\rank U\Sigma V^T = \rank X$ if the iterations are run all the way until $X_k=0$.
In general, $U\Sigma V^T \neq X$, with equality iff the residual $U_k^TX_k(I-V_kV_k^T)=0$ for all $k$.
Indeed, if equality holds, then $U^TX-\Sigma V^T=0$.
Step 3 of the algorithm above now implies that $U_k^TX=U_k^TX_k$ and \autoref{thmSubspaceSVD}, property \textit{4}, yields
\begin{align}
U^TX-\Sigma V^T =
\begin{pmatrix}
U_1^TX_1(I-V_1V_1^T) \\
U_2^TX_2(I-V_2V_2^T) \\
\vdots\\
U_n^TX_n(I-V_nV_n^T)
\end{pmatrix}.
\end{align}

To adaptively reduce noise, $\Pi$ must depend on $X$.
Our motivating example is to use $\Pi$ for selecting a subset of the variables that are likely to be less influenced by noise.
This is a special case of choosing $\Pi$ after performing a linear transform of the variables, which is described in the following theorem:
\begin{theorem}
\label{thmSMSSVD}
Take a linear map $S: \mathbb{R}^L \to \mathbb{R}^P$ and an integer $d$ such that $\rank S^TX\geq d$ and let $\tilde{U}\tilde{\Sigma}\tilde{V}^T$ be the rank $d$ truncated SVD of $S^TX$.
Furthermore let $\Pi$ be the subspace spanned by the columns of $\tilde{V}$ and let $U\Sigma V^T$ be the SVD of $\restr{X}{\Pi}$.
Then
\begin{enumerate}
	\item $\Pi \perp \ker X$.
	\item $S^TU\Sigma V^T = \tilde{U}\tilde{\Sigma}\tilde{V}^T$.
	\item $\{V_{\cdot 1}, V_{\cdot 2}, \hdots, V_{\cdot d} \}$ and $\{\tilde{V}_{\cdot 1}, \tilde{V}_{\cdot 2}, \hdots, \tilde{V}_{\cdot d} \}$ are orthonormal bases of $\Pi$.
	\item $\{S^TU_{\cdot 1}, S^TU_{\cdot 2}, \hdots, S^TU_{\cdot d} \}$ and $\{\tilde{U}_{\cdot 1}, \tilde{U}_{\cdot 2}, \hdots, \tilde{U}_{\cdot d} \}$ are bases of $S^TX(\Pi)$.
	\item $\|\Sigma\|_F \geq \frac{\|\tilde{\Sigma}\|_F}{\|S\|_2}$.
	\item $U^TX = \Sigma V^T + U^T(I-SS^T)X(I-VV^T)$.
\end{enumerate}
\end{theorem}
\begin{proof}
\textit{1.} The columns of $\tilde{V}$ are orthogonal to $\ker S^TX \supset \ker X$.
\textit{2.} $S^TU\Sigma V^T = S^T\restr{X}{\Pi} = \restr{(S^TX)}{\Pi} = \tilde{U}\tilde{\Sigma}\tilde{V}^T$.
\textit{3.} Follows immediately from the definitions.
\textit{4.} $\{\tilde{U}_{\cdot i}\}_{i=1}^d$ is a basis of $S^TX(\Pi)$.
By property \textit{2}, $\tilde{U} = S^TU\Sigma V^T\tilde{V}\tilde{\Sigma}^{-1}$, showing that $\{S^TU_{\cdot i}\}_{i=1}^d$ span $\{\tilde{U}_{\cdot i}\}_{i=1}^d$.
Finally, since $U$ and $\tilde{U}$ have the same rank, $\{U_{\cdot i}\}_{i=1}^d$ is also a basis of $S^TX(\Pi)$.
\textit{5.} For general matrices $A$ and $B$, consider $A$ acting on each column of $B$. 
We get
\[
\|AB\|_F^2=\sum_i \|AB_{\cdot i}\|_2^2 \leq \sum_i \|A\|_2^2\|B_{\cdot i}\|_2^2 = \|A\|_2^2\|B\|_F^2.
\]
The result now follows from property \textit{2}, with $A = S^T$ and $B = U\Sigma V^T$, since $\|AB\|_F=\|\tilde{U}\tilde{\Sigma}\tilde{V}^T\|_F=\|\tilde{\Sigma}\|_F$ and $\|B\|_F=\|\Sigma\|_F$.
\textit{6.} From \autoref{thmSubspaceSVD}, property \textit{4}, we get $U^TX = \Sigma V^T + U^TX(I-VV^T)$.
It remains to show that $U^TSS^TX(I-VV^T)=\bm{0}$.
By property \textit{4}, there exists a matrix $Z$ such that $S^TU=\tilde{U}Z$ and 
\begin{align}
U^TSS^TX(I-VV^T) &= Z^T\tilde{U}^TS^TX(I-VV^T)\\
&= Z^T\tilde{\Sigma}\tilde{V}^T(I-\tilde{V}\tilde{V}^T) = \bm{0},
\end{align}
where $VV^T=\tilde{V}\tilde{V}^T$ because of property \textit{3}.
\end{proof}

\begin{corollary}
If $S^TS=I$, then $\|\Sigma\|_F\geq\|\tilde{\Sigma}\|_F$.
\end{corollary}

Another way to interpret $S$ is that $SS^T$ defines a (possibly degenerate) inner product on the sample space, which is used to find $\Pi$.
To see this, let $d=\rank S^TX$ so that $\tilde{U}\tilde{\Sigma}\tilde{V}^T=S^TX$ and $K \coloneqq X^TSS^TX = \tilde{V}\tilde{\Sigma}^2\tilde{V}^T$, showing the well-known result that $\tilde{V}\tilde{\Sigma}^2\tilde{V}^T$ is an eigendecomposition of $K$, where $K_{ij} = \langle x_i, x_j \rangle \coloneqq X_{\cdot i}^TSS^TX_{\cdot j}$ is the inner product of sample $i$ and $j$.
This naturally extends to kernel PCA, where $K$ is defined by taking scalar products after an (implicit) mapping to a higher-dimensional space.
Any method that results in a low-dimensional sample space representation can indeed be used, since $\Pi$ is spanned by the columns of $V$ by definition.
We will not pursue these extensions here.

The Projection Score \cite{fontes2011projection} provides a natural optimality criterion for $S$ and $d$ (and thus $\Pi$) needed in each iteration of the SMSSVD algorithm.
It is a measure of how informative a specific variable subset is, when constructing a rank $d$ approximation of a data matrix.
A common application is to maximize the Projection Score over a sequence of variable subsets, where each subset consists of those variables that have a variance above a specific threshold.
Using the notation from \autoref{thmSMSSVD}, the optimal variable subset describes a matrix $S$ and the optimal low-rank approximation is $\tilde{U}\tilde{\Sigma}\tilde{V}^T$.
Here $S$ has exactly one element in each column equal to $1$, at most one element in each row equal to $1$ and all other elements equal to zero.
Hence $S^TX$ corresponds to selecting a subset of the variables of a data matrix $X$ and $S^TS=I$.
In iteration $k$ of the SMSSVD algorithm, we optimize the Projection Score jointly over the variance filtering threshold and the dimension, which gives both an optimal variable subset $S_k$ and a simple dimension estimate $d_k$ of the signal that was captured.

\section{Results}
The performance of SMSSVD is evaluated in comparison to SVD and SPC (Sparse Principal Components), a method similar to SVD, but with an additional lasso ($L_1$) constraint to achieve sparsity \cite{witten2009penalized}.
The methods are evaluated both for real data using three Gene Expression data sets and for synthetic data where the ground truth is known.

\subsection{Gene Expression Data}
Three Gene Expression data sets, two openly available with microarray data and one based on RNA-Seq available upon request from the original authors, were analyzed.
Gene expression microarray profiles from a study of breast cancer \cite{chin2006genomic} was previously used to evaluate SPC \cite{witten2009penalized}, but in contrast to their analysis, we use all $118$ samples and all $22215$ genes.
Each sample was labeled as one of five breast cancer subtypes: `basal-like', `luminal A', `luminal B', `ERBB2', and `normal breast-like'.
In a study of pediatric Acute Lymphoblastic Leukemia (ALL), gene expression profiles were measured for $132$ diagnostic samples \cite{ross2003classification}.
The samples were labeled by prognostic leukemia subtype (`TEL-AML1', `BCR-ABL', `MLL', `Hyperdiploid (>50)', `E2A-PBX1', `T-ALL' and `Other').
Our final data set is from another pediatric ALL study, where gene expression profiling was done from RNA-Seq data for $195$ samples \cite{lilljebjorn2016identification}.
The samples were aligned with Tophat2 \cite{kim2013tophat2} and gene expression levels were normalized by TMM \cite{robinson2010scaling}.
Only genes with a support of at least $10$ reads in at least $2$ samples were kept.
The annotated subtypes in this data set were `BCR-ABL1', `ETV6-RUNX1', `High hyperdiploid', `MLL', `TCF3-PBX1' and `Other'.
Here, `Other' is a very diverse group containing everything that did not fit in first five categories.
We thus present results both with and without this group included.

The ability to extract relevant information from the gene expression data sets was evaluated for each model by how well they could explain the subtypes, using the Akaike Information Criterion (AIC) for model scoring.
Given the low-dimensional sample representations from SMSSVD, SVD or SPC (for different values of the sparsity parameter, $c$), a Gaussian Mixture Model was constructed by fitting one Multivariate Gaussian per subtype.
The class priors were chosen proportional to the size of each subtype.
The loglikelihood $l \coloneqq \log P(\bm{x}|\theta, M)$, where $\bm{x}$ are the subtype labels, $M$ is the model and $\theta$ a vector of $k$ fitted model parameters is used to compute the $\text{AIC} = 2k-2l$.
\autoref{figCancers} displays the AIC scores for the different models as a function of the model dimension.
SMSSVD generally performs better than SVD, by a margin.
Comparison with SPC is trickier, since the performance of SPC is determined by the sparsity parameter $c$ and there is no simple objective way to choose $c$.
However, SMSSVD compares well with SPC regardless of the value of the parameter.

\begin{figure}[!htpb]
  \centering
  \includegraphics[trim=0mm 5mm 0mm 2mm, clip, scale=0.4]{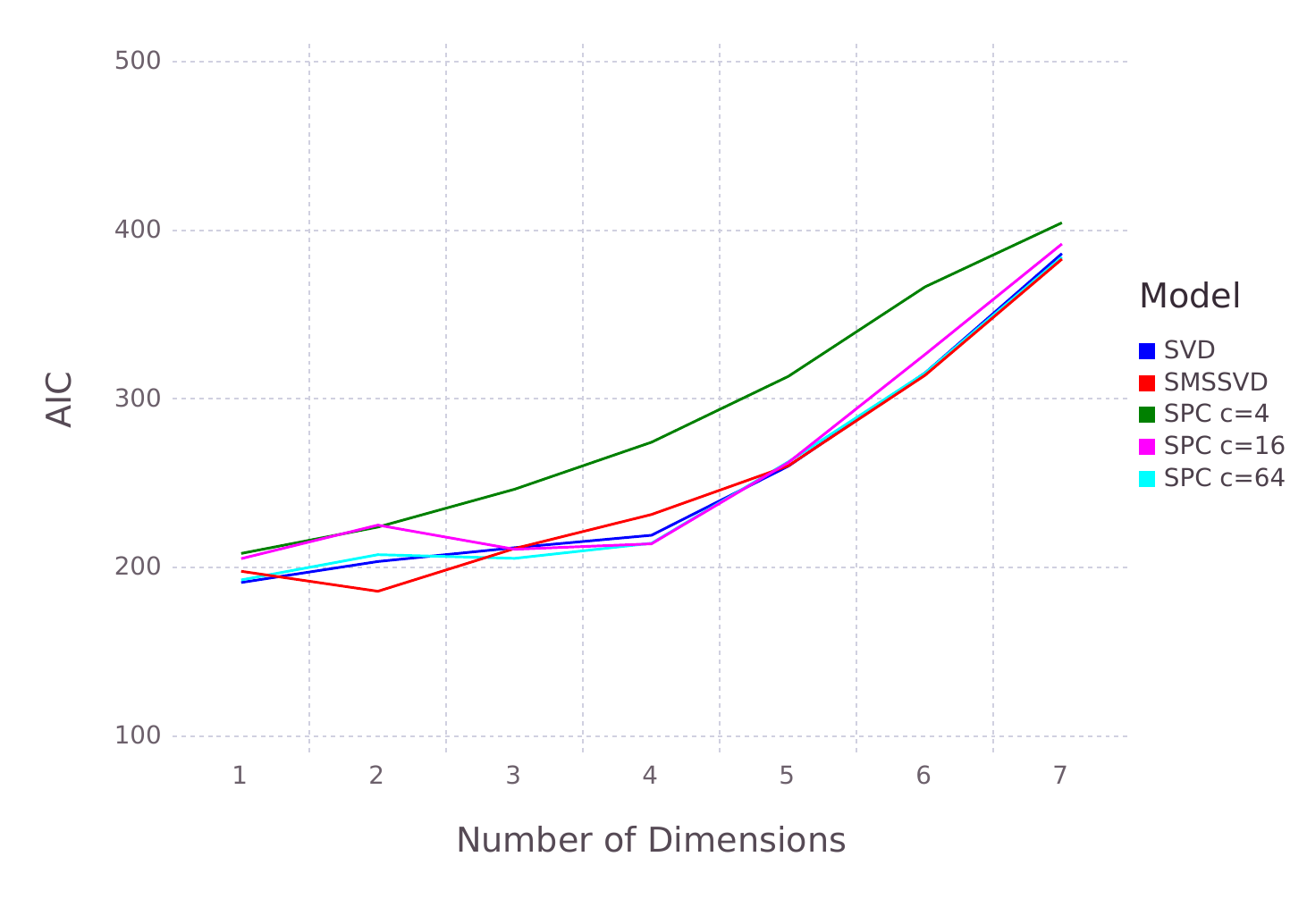}
  \includegraphics[trim=0mm 5mm 0mm 2mm, clip, scale=0.4]{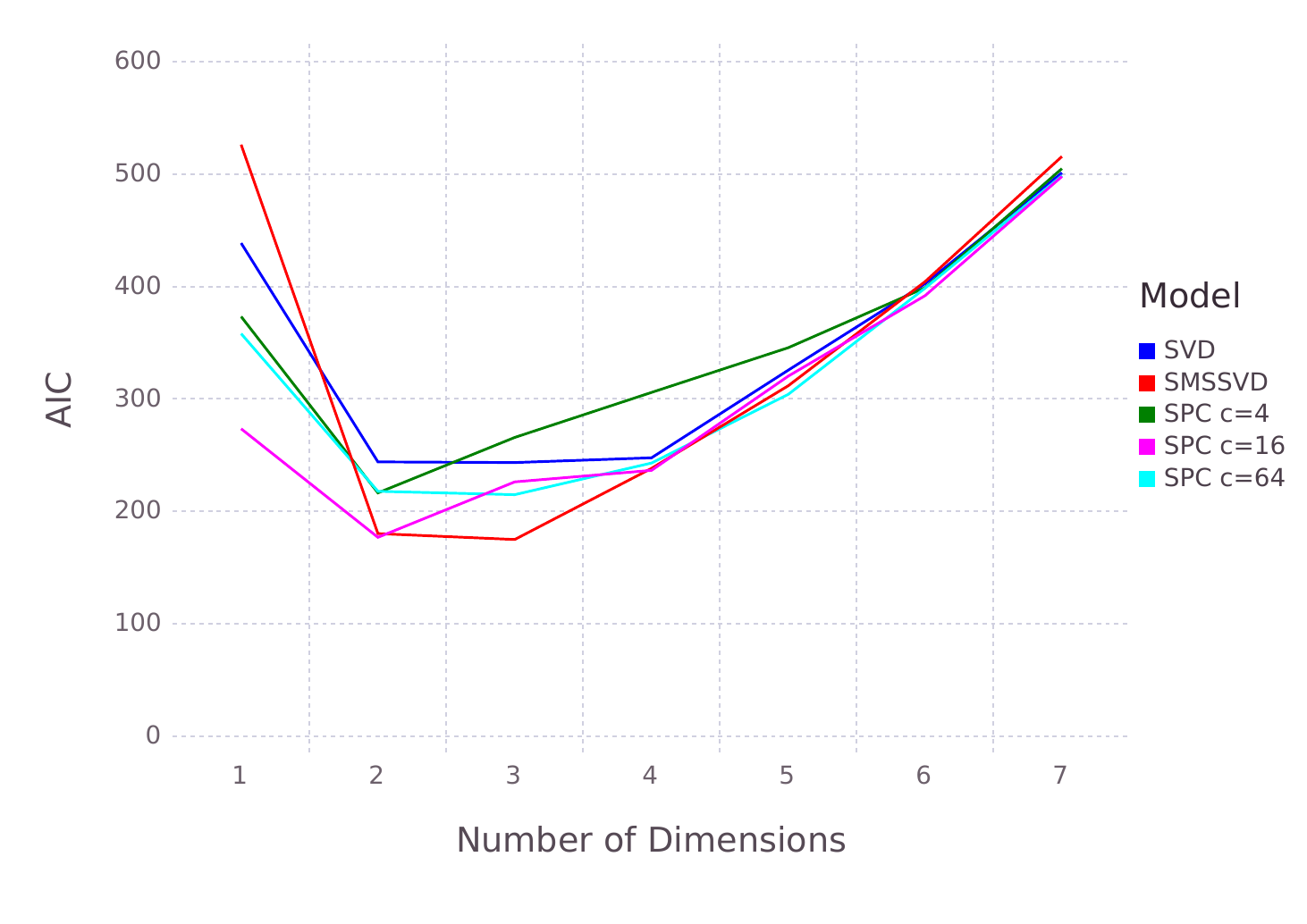}
  \includegraphics[trim=0mm 5mm 0mm 2mm, clip, scale=0.4]{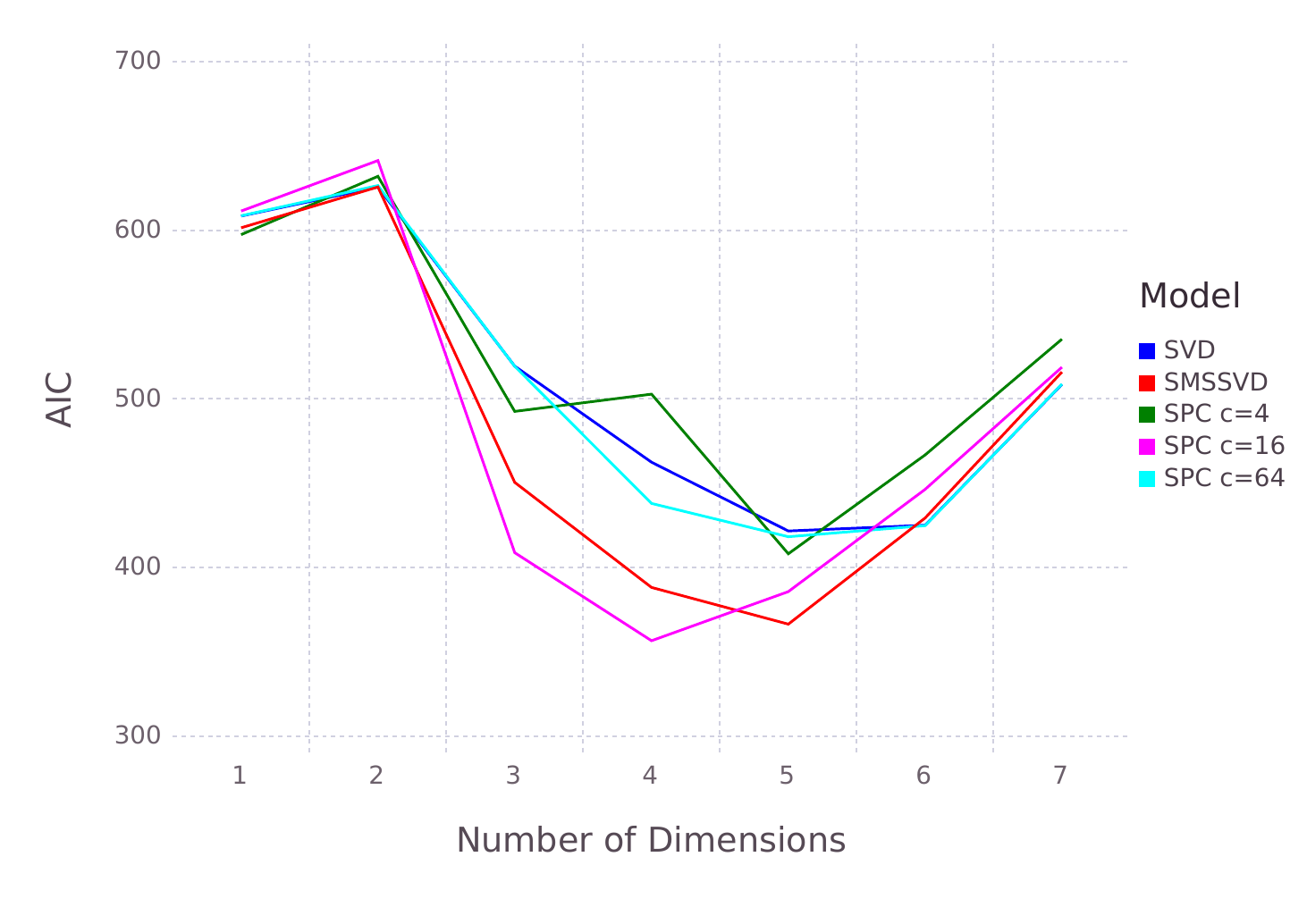}
  \includegraphics[trim=0mm 5mm 0mm 2mm, clip, scale=0.4]{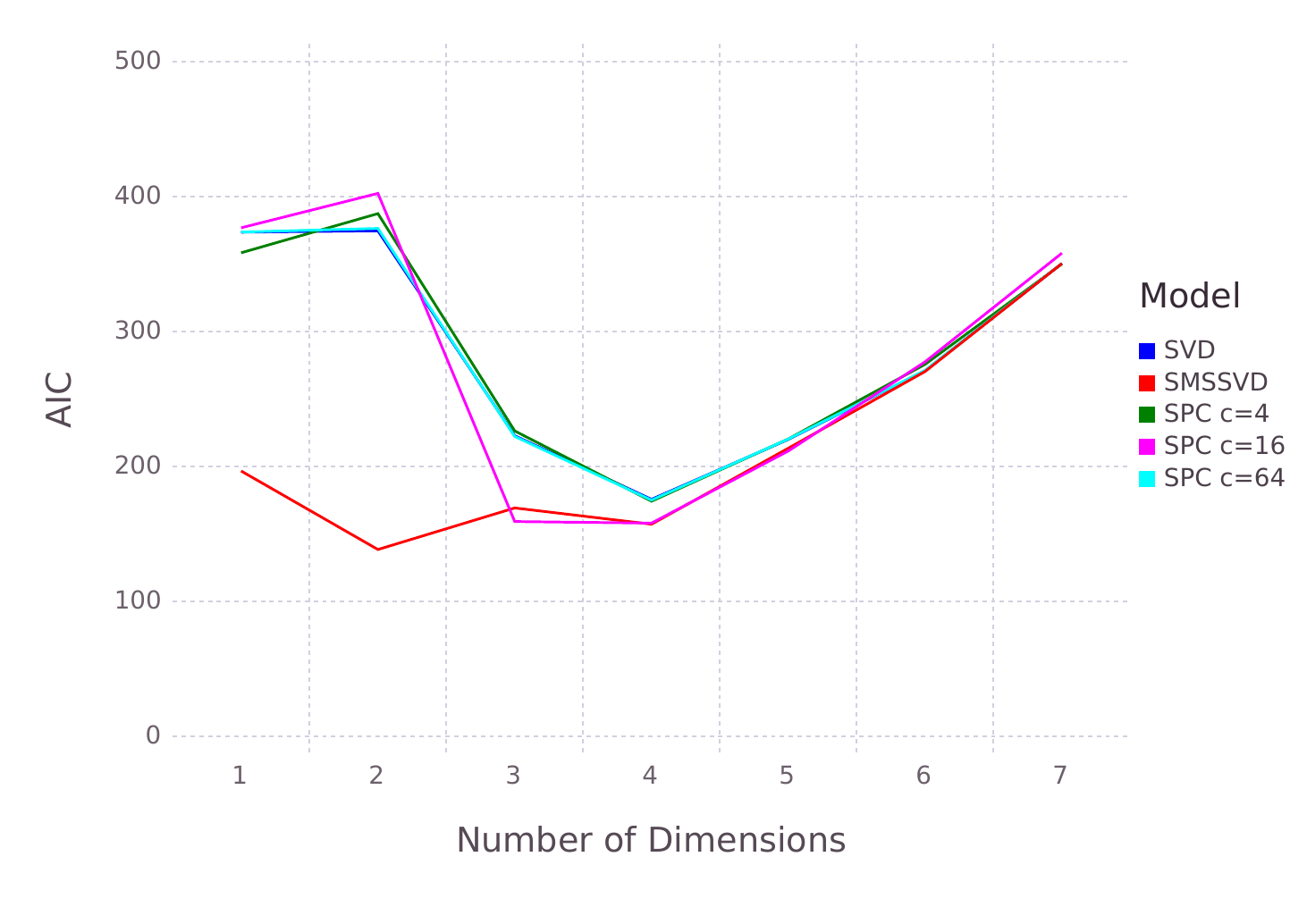}
  \caption{Evaluation of SMSSVD on different data sets, based on AIC scores when fitting a Gaussian Mixture Model to the subtypes. From top to bottom: A. Breast Cancer, B. Acute Lymphoblastic Leukemia (Microarray), C. Acute Lymphoblastic Leukemia (RNA-Seq), D. Acute Lymphoblastic Leukemia (RNA-Seq) with subtype `Other' removed.}
  \label{figCancers}
\end{figure}

\subsection{Synthetic Data}

SMSSVD decomposes a matrix observed in noisy conditions as a series of orthogonal low-rank signals.
The aim is to get a stable representation of the samples and then recover as much as possible of the variables, even for signals that are heavily corrupted by noise.
To evaluate SMSSVD, we synthetically create a series of low-rank signals $Y_k$ that are orthogonal (i.e. $Y_i^TY_j=0$ and $Y_iY_j^T=0$ for $i\neq j$) and that has a chosen level of sparsity on the variable side and try to recover the individual $Y_k$'s from the observed matrix $X\coloneqq\sum_kY_k + \varepsilon$ where $\varepsilon$ is a matrix and $\varepsilon_{ij} \sim \mathcal{N}(0,\sigma_{ij})$.
To measure how well SMSSVD recovers the signals from the data, we look at each signal separately, considering only variables where the signal has support.
Let $\text{err}(k)$ be the reconstruction error of signal $k$, 
\[
\text{err}(k) \coloneqq \|R_k^T(Y_k - \hat{Y}_k)\|_F,
\]
where $\hat{Y}_k$ is the reconstructed signal and $R_k$ is defined such that multiplying with $R_k^T$ from the left selects the variables (rows) where $Y_k$ is nonzero. 

While SMSSVD is designed to find $d$-dimensional signals ($\hat{Y}_k\coloneqq U_k\Sigma_kV_k^T$), the same is not true for SVD and SPC.
To test the ability to find the signals, rather than the ability to find them in the right order, the components are reordered using a algorithm that tries to minimize the total error by greedily matching the rank $1$ matrices from the decomposition to signals $Y_k$, always picking the match that lowers the total error the most.
The number of rank $1$ matrices matched to each signal $Y_k$ is equal to $\rank Y_k$.
Note that with no noise present, SVD is guaranteed to always find the optimal decomposition.

The biplots in \autoref{figSMSSVDBiplots} illustrate how SMSSVD works and how the signal reconstructions compares to other methods.
If there is no noise, perfect decompositions are achieved by all methods apart from SPC with a high degree of sparsity.
An artificial example where the noise is only added to the non-signal variables highlights that SMSSVD can still perfectly reconstruct both samples and signal variables, whereas the other methods display significant defects.
Finally, when all variables are affected by noise, SMSSVD still get the best results.

\begin{figure*}[!htpb]
	\centering
	\includegraphics[scale=0.4]{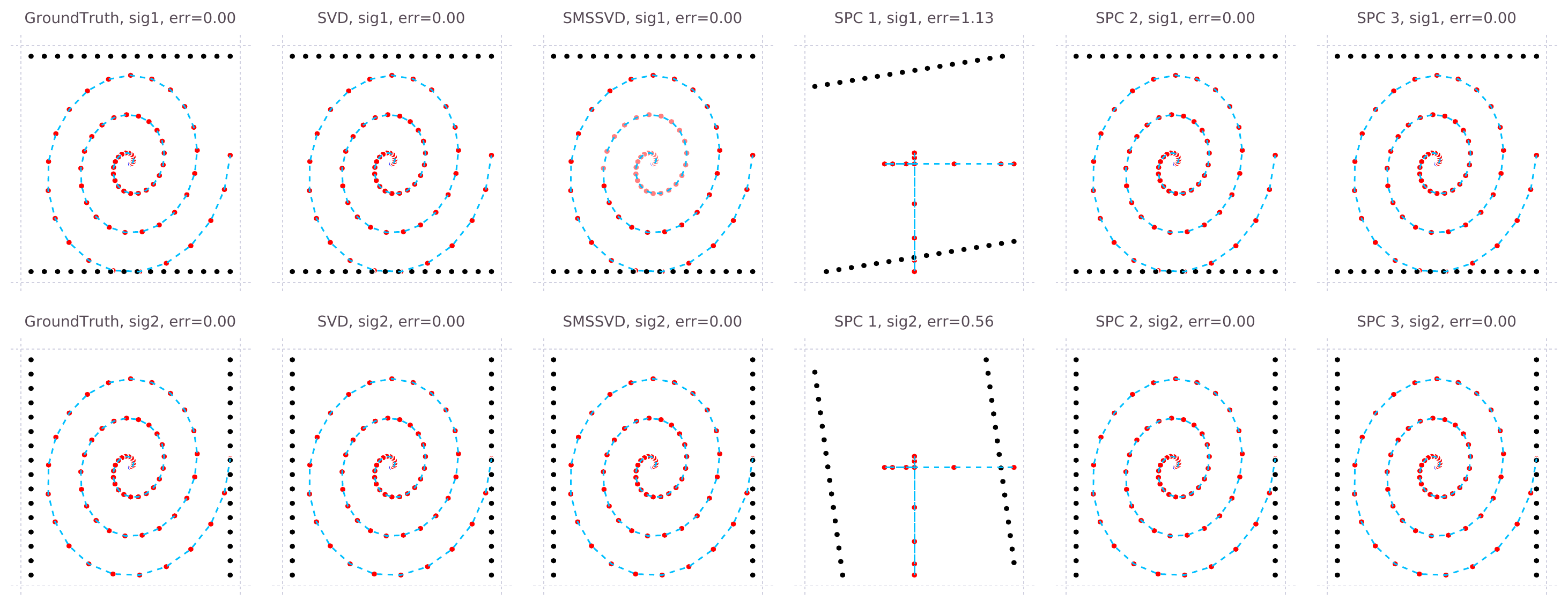}
	\vspace{3mm}
	\includegraphics[scale=0.4]{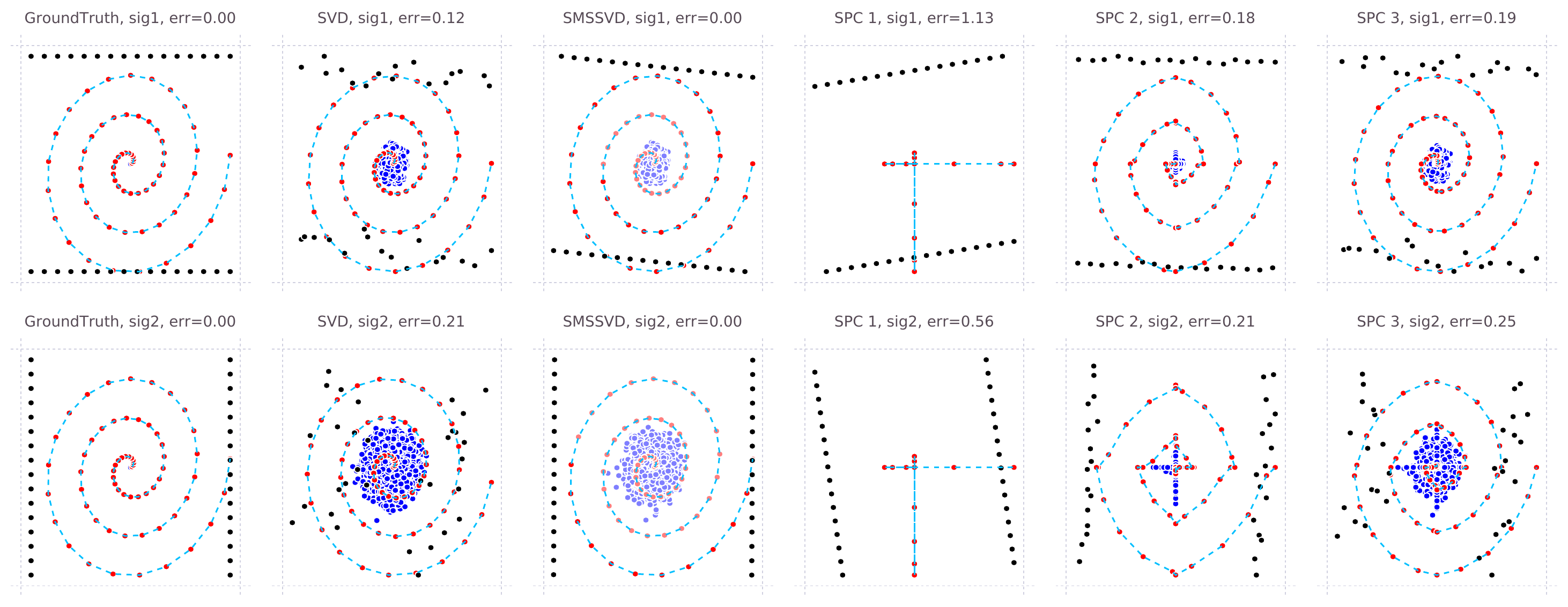}
	\vspace{3mm}
	\includegraphics[scale=0.4]{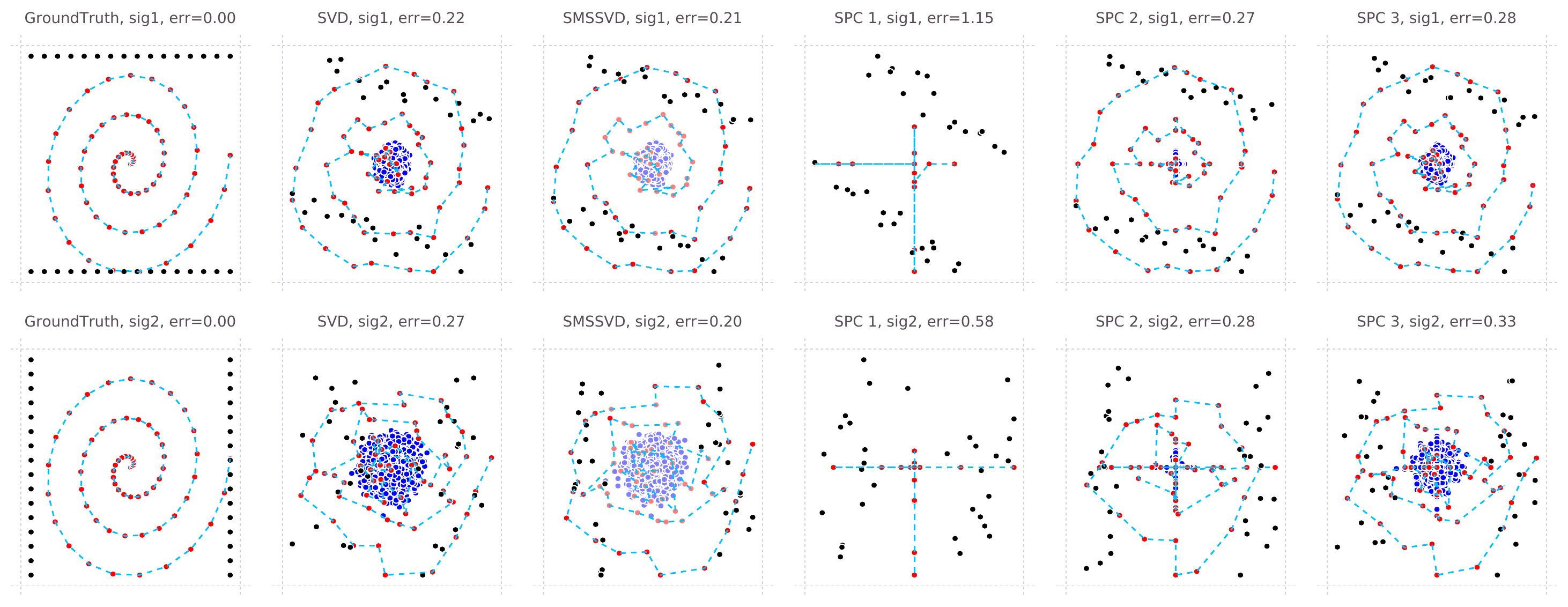}
	\caption{Two $2d$ signals with non-overlapping support for the variables are shown for no noise (Upper two rows), noise on non-signal variables only (Middle two rows) and for noise on all variables (Lower two rows). The reconstruction of the first signal is shown in the upper row and for the second signal in the lower row in each set. Different columns correspond to different methods, where SPC ``1'', ``2'' and ``3'' have regularization penalties of $c=2$, $8$ and $32$ respectively, controlling the degree of sparsity. Samples are black, variables where the signal has support are red and other variables are blue. The variables in the support are connected with dashed lines, only to make it easier to spot how the variables are influenced by noise. For SMSSVD, variables selected by optimal variance filtering are shown in full color and other variables are shown in a whiter tone. Samples and variables are both scaled to fill the axes in each biplot. $32$ samples and $5000$ variables were used, of which each signal had support in $64$ variables and the rest had noise only.}
	\label{figSMSSVDBiplots}
\end{figure*}

Next, we created several data sets for a variety of conditions based on the parameters $N=100$: Number of samples, $P$: Number of variables, $L$: Number of variables in the support of each signal, $K=8$: number of signals and $d$: the rank of each signal.
For each signal, we randomize matrices $U_k$ and $V_k$, choose a diagonal matrix $\Sigma_k$ and let $Y_k\coloneqq U_k\Sigma_kV_k^T$.
For both $V_k$ and $U_k$, each new column is created by sampling a vector of i.i.d. Gaussian random variables and projecting onto the orthogonal complement of the subspace spanned by previous columns (in current and previous signals).
For $U_k$, we only consider the subspace spanned by $L$ randomly selected variables.
The result is then expanded by inserting zeros for the other $P-L$ variables.
To complete the signal, let the $i$'th diagonal element of $\Sigma_k$, $(\Sigma_k)_{ii} \coloneqq 0.6^{k-1}0.9^{i-1}$, such that there is a decline in the power between signals and within components of each signal.
Finally, i.i.d. Gaussian noise is added to the data matrix.
\Autoref{figSMSSVDResiduals1,figSMSSVDResiduals2,figSMSSVDResiduals3} show test results for data sets randomized in this way for different sets of parameters.
SMSSVD is the only method that performs well over the whole set of parameters.
The only situation where SMSSVD is consistently outperformed is by SVD for large $L$, and it is by a narrow margin.
SMSSVD performs particularly well, in comparison to the other methods, in the difficult cases when the signal to noise ratio is low.
SPC performance clearly depends on the regularization parameter which must be chosen differently in different situations.
However, despite being a parameter-free method, SMSSVD outperforms SPC in most cases.

\begin{figure*}[!htpb]
    \centering
    \includegraphics[width=\textwidth]{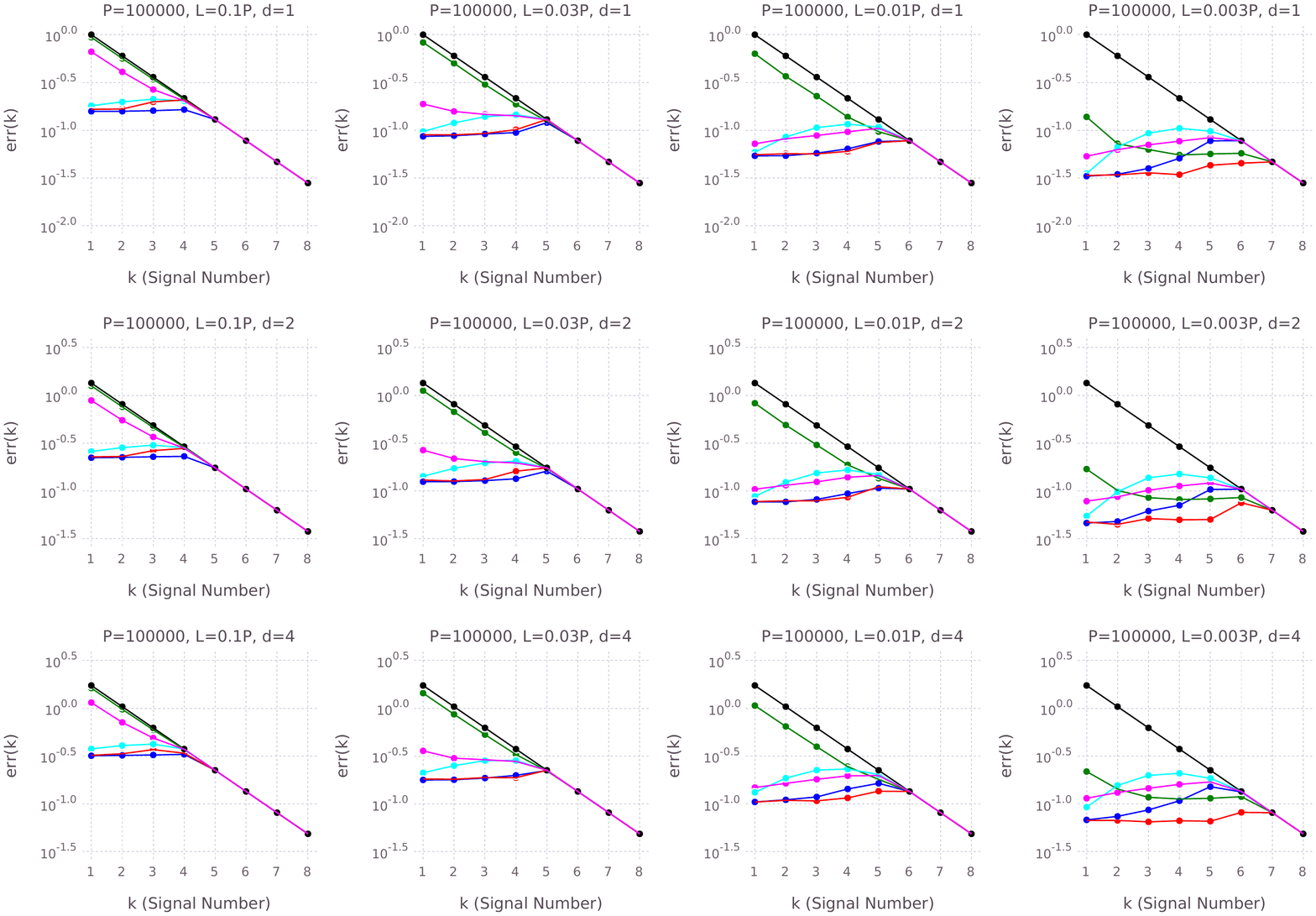}
    \caption{The reconstruction error, $\text{err}(k)$, is shown for different conditions. The signal strength $\|Y_k\|_F$ (black) is shown for scale. The methods are: SVD (blue), SMSSVD (red) and SPC (green, magenta, cyan) with decreasing degree of sparsity (regularization parameters $c=0.04\sqrt{P}$, $c=0.12\sqrt{P}$ and $c=0.36\sqrt{P}$ respectively). No errors larger than the signal strength are displayed as that indicates that a different signal has been found.}
  \label{figSMSSVDResiduals1}
\end{figure*}

\begin{figure*}[!htpb]
    \centering
    \includegraphics[width=\textwidth]{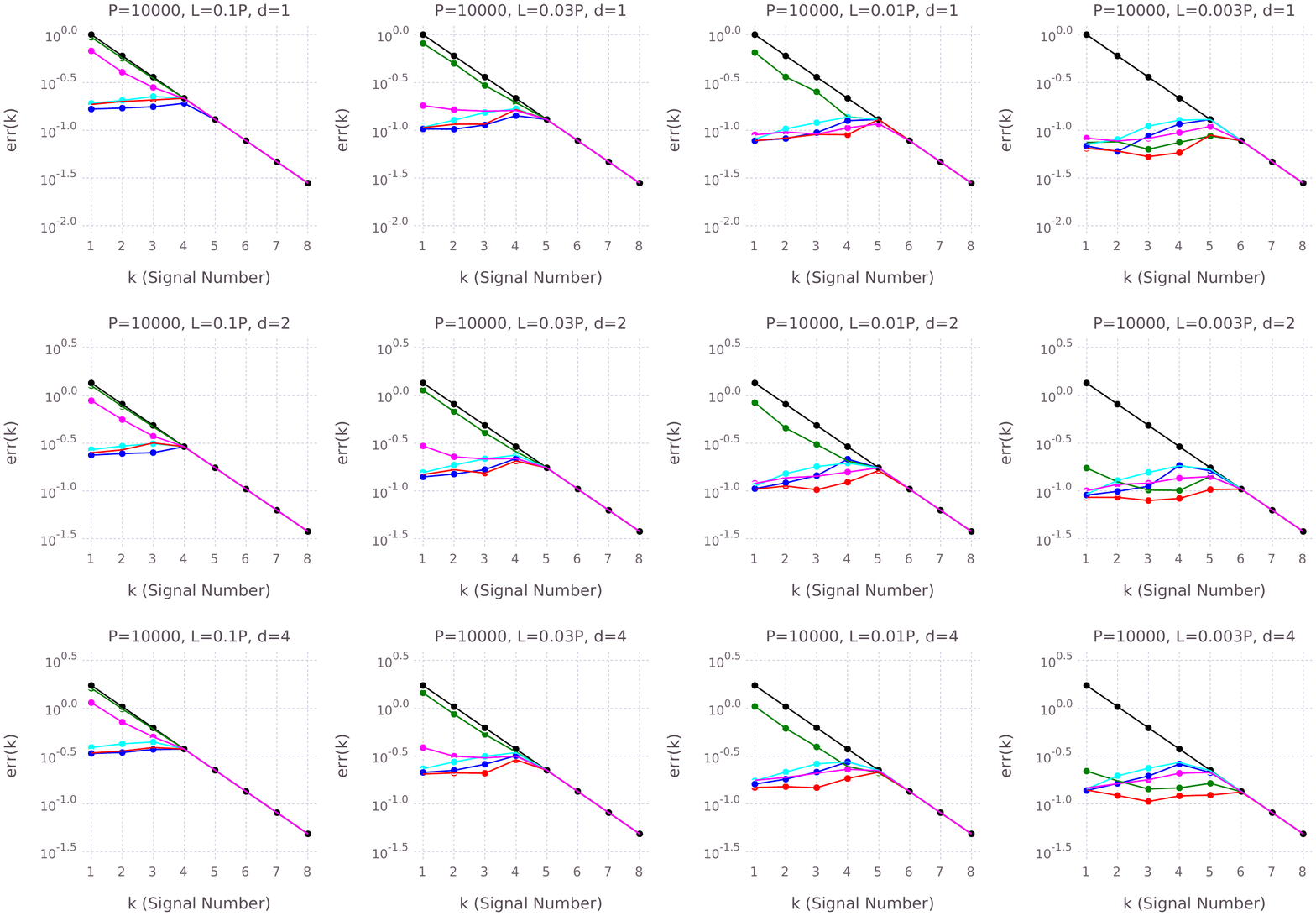}
    \caption{The reconstruction error, $\text{err}(k)$, is shown for different conditions. The signal strength $\|Y_k\|_F$ (black) is shown for scale. The methods are: SVD (blue), SMSSVD (red) and SPC (green, magenta, cyan) with decreasing degree of sparsity (regularization parameters $c=0.04\sqrt{P}$, $c=0.12\sqrt{P}$ and $c=0.36\sqrt{P}$ respectively). No errors larger than the signal strength are displayed as that indicates that a different signal has been found.}
  \label{figSMSSVDResiduals2}
\end{figure*}

\begin{figure*}[!htpb]
    \centering
    \includegraphics[width=\textwidth]{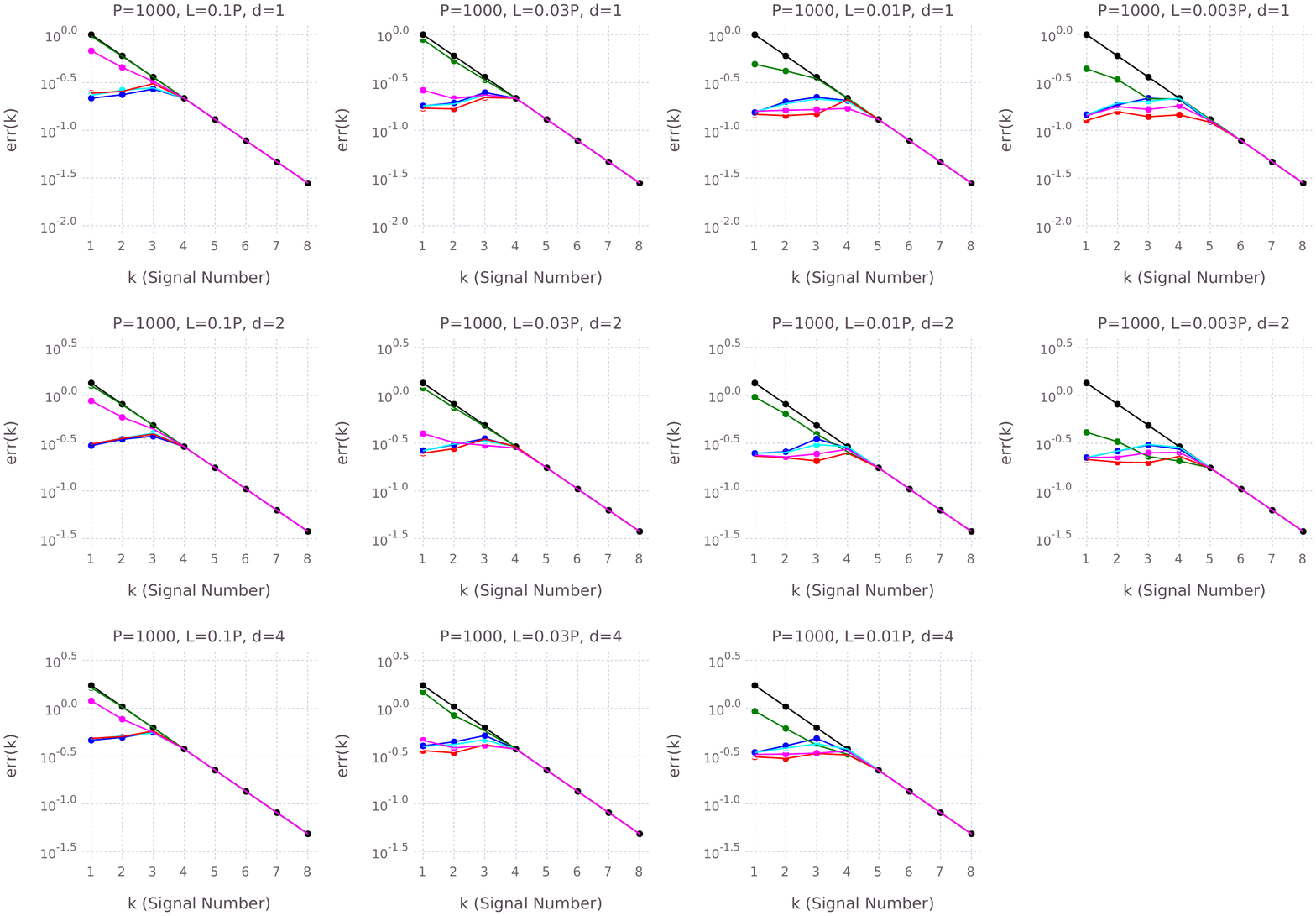}
    \caption{The reconstruction error, $\text{err}(k)$, is shown for different conditions. The signal strength $\|Y_k\|_F$ (black) is shown for scale. The methods are: SVD (blue), SMSSVD (red) and SPC (green, magenta, cyan) with decreasing degree of sparsity (regularization parameters $c=0.04\sqrt{P}$, $c=0.12\sqrt{P}$ and $c=0.36\sqrt{P}$ respectively). No errors larger than the signal strength are displayed as that indicates that a different signal has been found.}
  \label{figSMSSVDResiduals3}
\end{figure*}

\section{Discussion}
We have presented SMSSVD, a dimension reduction technique designed for complex data sets with multiple overlaid signals observed in noisy conditions.
When compared to other methods, over a wide range of conditions, SMSSVD performs equally well or better.
SMSSVD excels in situations where $P\gg N$ (many more variables than samples) but most of the variables just contribute with noise, a very common situation for high throughput biological data.
As a parameter-free method, SMSSVD requires no assumptions to be made of the level of sparsity.
Indeed, SMSSVD can handle different signals within the same data set that exhibit very different levels of sparsity.
Being parameter-free also makes SMSSVD suitable for automated pipelines, where few assumptions can be made about the data.

A common strategy when analyzing high dimensional data is to first apply PCA (SVD) to reduce the dimension to an intermediate number, high enough to give an accurate representation of the data set, but low enough to get rid of some noise and to speed up downstream computations (see e.g. \cite{maaten2008visualizing}).
We argue that since SMSSVD can recover multiple overlaid signals and adaptively reduce the noise affecting each signal so that even signals with a lower signal to noise ratio can be found, it is very useful in this situation.

Our unique contribution is that we first solve a more suitable dimension reduction problem for robustly finding signals in a data set corrupted by noise and then map the result back to the original variables.
We also show how this combination of steps gives SMSSVD many desirable properties, related to the SVD of both the full data matrix and of the smaller matrix from the variable selection step.
Orthogonality between components is one of the cornerstones of SVD, but it is often difficult to satisfy the orthogonality conditions when other factors are taken into account.
SPC does for instance give orthogonality for samples, but not for variables and the average genes of each subset in gene shaving are `reasonable' uncorrelated. 
For SMSSVD, orthogonality follows immediately from the construction, simplifying interpretation and subsequent analysis steps.
\autoref{thmSMSSVD}, property \textit{2} highlights that the variables retained in the variable selection step are unaffected when the solution is expanded to the full set of variables.
Hence, we can naturally view each signal from the point of view of the selected variables, or using all variables.

The variable selection step in the SMSSVD algorithm can be chosen freely.
For exploratory analysis, optimizing the Projection Score based on variance filtering is a natural and unbiased choice.
Another option is to use Projection Score for response related filtering, e.g. ranking the variables by the absolute value of the $t$-statistic when performing a $t$-test between two groups of samples.
The algorithm also has verbatim support for variable weighting, by choosing the $S$ matrix as a diagonal matrix with a weight for each variable.
Clearly this is a generalization of variable selection.

Kernel PCA, SPC, and other methods that give low-dimensional sample representations, but where the variable information is (partially) lost, can also be extended by SMSSVD (relying on \autoref{thmSubspaceSVD} only), as long as a linear representation in the original variables can be considered meaningful.
Apart from retrieving a variable-side representation, the SMSSVD algorithm also makes it possible to find multiple overlapping signals, by applying the dimension reduction method of interest as the first step of each SMSSVD iteration.

\section*{Acknowledgements}
The authors would like to thank Thoas Fioretos and Henrik Lilljebj\"orn at the Department of Laboratory Medicine, Division of Clinical Genetics, Lund University, for giving us access to the RNA-seq data presented in \cite{lilljebjorn2016identification}.

\printbibliography

\end{document}